\documentclass[A4,11pt]{article}
\usepackage{microtype}
\usepackage[utf8]{inputenc}
\usepackage{geometry}
\usepackage{rotating}
\usepackage{amsmath}
\usepackage{amsfonts}
\usepackage{multirow}
\usepackage{tikz}
\usepackage{amsthm}
\usepackage{mathtools}
\usepackage{comment}

\usepackage{booktabs}

\newtheorem{theorem}{Theorem}[section]
\newtheorem*{theorem*}{Theorem}
\newtheorem{corollary}[theorem]{Corollary}

\newtheorem{problem}{Problem}

\newtheorem{lemma}[theorem]{Lemma}

\newtheorem{remark}[theorem]{Remark}

\usepackage{xcolor}
\usepackage{bbm}
\usepackage{algorithm,algorithmicx}
\usepackage[noend]{algpseudocode}

\definecolor{mygreen}{RGB}{0,128,0}

\definecolor{mysilver}{RGB}{220,220,220}

\def\avg{\text{AvgRank}}
\def\one{{\mathbf{1}}}

\def\R{{\mathcal R}}

\def\ZZ{{\mathbb Z}}

\def\P{{\mathcal P}}

\def\X{{\mathcal X}}
\def\C{{\mathcal H}}

\def\Z{{\mathcal Z}}
\def\S{{\mathcal R}}
\def\E{{\mathcal E}}

\def\H{{\mathcal H}}
\def\B{{\mathcal B}}
\def\V{{\mathcal V}}

\def\W{{\mathcal W}}

\newcommand{\CAexp}{\textsc{Min-Avg$_{\text{exp}}$ HR}}
\newcommand{\CAexpOpt}{\textsc{Min-Avg$_{\text{exp}}$ HR opt}}
\newcommand{\CAred}{\textsc{Min-Avg$_{\text{red}}$ HR}}
\newcommand{\CAredOpt}{\textsc{Min-Avg$_{\text{red}}$ HR opt}}

\newcommand{\CAexpsub}{\textsc{Min-Avg$_{\text{exp}}^{\text{sub}}$ HR}}
\newcommand{\CAexpsubOpt}{\textsc{Min-Avg$_{\text{exp}}^{\text{sub}}$ HR opt}}
\newcommand{\SMexpsub}{\textsc{Min-Avg$_{\text{exp}}^{\text{sub}}$ SM}}

\newcommand{\CAredsub}{\textsc{Min-Avg$_{\text{red}}^{\text{sub}}$ HR}}
\newcommand{\CAredsubOpt}{\textsc{Min-Avg$_{\text{red}}^{\text{sub}}$ HR opt}}

\newcommand{\Maxcard}{\textsc{Max-Card HRTI}}

\newcommand{\Maxcardexp}{\textsc{Max-Card$_{\text{exp}}^{\text{sub}}$ HRI}}

\newcommand{\Maxcardred}{\textsc{Max-Card$_{\text{red}}^{\text{sub}}$ SMI}}

\newcommand{\MinW}{\textsc{Min-w SMT}}
\newcommand{\MinWOpt}{\textsc{Min-w SMT opt}}

\newcommand{\CA}{CA}

\newcommand{\ie}{\emph{i.e.}}
\newcommand{\eg}{\emph{e.g.}}

\usepackage{enumerate}
\newcommand{\boxxx}[1]
 {\begin{center}\fbox{\begin{minipage}{40em}#1\smallskip\end{minipage}}\end{center}}

\usepackage{xr}
\usepackage{subcaption}
\usepackage{mathtools}

\begin{document}

	\title{Capacity Variation in the Many-to-one Stable Matching\\ \footnotesize{To  Gerhard Woeginger (1964-2022), an outstanding computer scientist}}
	
	\author{Federico Bobbio
	\thanks{CIRRELT and DIRO, Université de Montr\'eal. CERC, Polytechnique Montr\'eal. {\tt federico.bobbio@umontreal.ca}}
	\and Margarida Carvalho
	\thanks{CIRRELT and DIRO, Universit\'e de Montr\'eal. {\tt carvalho@iro.umontreal.ca}}
	\and Andrea Lodi
	\thanks{Jacobs Technion-Cornell Institute, Cornell Tech. {\tt andrea.lodi@cornell.edu}}
	\and 	Alfredo Torrico		
	\thanks{CERC Data Science, Polytechnique Montr\'eal. {\tt alfredo.torrico-palacios@polymtl.ca}}
	}
\date{\vspace{-1em}}

\maketitle
\begin{abstract}
    The many-to-one stable matching problem provides the fundamental abstraction of several real-world matching markets such as  school choice and hospital-resident allocation. The agents on both sides are often referred to as residents and hospitals. The classical setup assumes that the agents rank the opposite side and that the capacities of the hospitals are fixed. 


It is known that increasing the capacity of a single hospital improves the residents' final allocation. On the other hand, reducing the capacity of a single hospital deteriorates the residents' allocation. In this work, we study the computational complexity of finding the optimal variation of hospitals' capacities that leads to the best outcome for the residents, subject to stability and a capacity variation constraint.

First, we show that the decision problem of finding the optimal capacity expansion is NP-complete and the corresponding optimization problem is inapproximable within a certain factor. This result holds under strict and complete preferences, and even if we allocate extra capacities to disjoint sets of hospitals.
Second, we obtain analogous computational complexity results for the problem of capacity reduction.
Finally, we study the variants of these problems when the goal is to maximize the size of the final matching under incomplete preference lists.

\end{abstract}
\thispagestyle{empty}
\newpage

\maketitle

\algrenewcommand\algorithmicrequire{\textbf{Input:}}
\algrenewcommand\algorithmicensure{\textbf{Output:}}

\allowdisplaybreaks

\section{Introduction}\label{sec:introduction}




The stable matching problem has found multiple applications such as daycare admission in Denmark \cite{kennes2011daycare}, school and hospital-resident allocation in the USA \cite{abdulkadirouglu2005boston, abdulkadiroglu2006changing, abdulkadirouglu2005new,roth1984evolution,roth1999redesign}, school and university admission in Hungary \cite{biro2008student, biro2015college}, school admission in Singapore \cite{teo1999gale}, university admission in China \cite{zhang2010analysis}, Germany \cite{braun2010telling} and Spain \cite{romero1998implementation}, faculty recruitment in France \cite{baiou2003admissions}. The many-to-one stable matching problem (HR) consists of two sides---henceforth referred to as hospitals and residents---where hospitals have fixed and known capacities. Both sides have preferences over each other, and the goal of the decision-maker is to find an assignment such that, in each pair, both agents \emph{simultaneously} prefer each other over any other agent. The HR problem, and its multiple variants, have been widely studied in the literature by different disciplines: From a polyhedral  \cite{baiou2000stable, baiou2004student} and algorithmic  \cite{gale1962college} perspective, to geometry \cite{sethuraman2006many}, mathematical programming \cite{vate1989linear}, combinatorics \cite{knuth1976marriages}, fixed-point methods \cite{subramanian1994new} and graph theory \cite{balinski1997stable}. 

As mentioned, in the standard version of HR, the capacity of the hospitals are fixed and known in advance. The decision-maker in charge of the final assignment does not have control over these quota. However, there are multiple real-life situations in which the variation of the size of the market, expansion or reduction, could play a significant role. For example, when allocating couples in hospitals~\cite{roth1984evolution}, siblings in school~\cite{correa2021school}, scholarships or expenses reduction.

The idea of introducing new participants in the matching market has been previously studied through the lens of game theory and econonomics. 
This problem is known as \textit{entry comparative static}, and is usually assumed that the introduced agent is an independent entity with a certain preference list; the participants of the opposite side also rank this new agent. It has been shown that when a new agent is introduced, then the resulting matching is \textit{weakly better} (\ie, equal or better) for the agents of the opposite side  \cite{KelsoCrawford1982,GaleSotomayor1985,RothSotomayor1990}. On the mathematical programming domain, for the HR, the problem of deciding simultaneously capacity expansions  on the hospitals' side and a stable matching was first proposed in \cite{bobbio2021capacity}. Using integer programming, the authors demonstrated empirically that significantly better matchings for the residents can be obtained through the allocation of a few extra spots.

In this work, we study the computational complexity of the problem proposed in~\cite{bobbio2021capacity} as well as its counterpart, \ie, when reduction of the hospitals' capacity is required. Roughly speaking, for the expansion of the market we study the following question: 
\begin{quote}
    \it Given a non-negative integer number $B\in\ZZ_+$ of extra spots, which hospitals should the decision-maker expand the capacity of to obtain the best stable matching for the residents?
\end{quote}
In the second part of this work, we focus on the reduction of the market. Simply, we study the following question:
\begin{quote}
\it Given a non-negative integer number $B\in\ZZ_+$ of spots to be removed, which hospitals should the decision-maker reduce the capacity of in order to obtain the best stable matching for the residents?
\end{quote}

We primarily focus on a rank-based metric to choose the best matching for residents. We also study the variants of the problems above under a cardinality-based metric, which has been widely studied in the literature \cite{roth1984evolution,GaleSotomayor1985,roth1986allocation,manlove2013algorithmics}.

\subsection{Related Work}\label{sec:related_work}

In their seminal paper, Gale and Shapley \cite{gale1962college} introduced the stable matching problem and provided a polynomial time algorithm known as the deferred acceptance (DA) algorithm. The DA algorithm computes an assignment such that there is no pair of agents that would simultaneously prefer to be paired to each other rather than being in their current assignment; this is known as a \emph{stable matching}. In practice, the DA mechanism has been extensively used to improve admission processes, \eg, see~\cite{abdulkadirouglu2005new,biro2008student}. For further details on stable matching mechanisms, see~\cite{RothSotomayor1990,manlove2013algorithmics}.
In general, the main focus of the literature has been on finding the maximum cardinality stable matching, which can be efficiently obtained when there are incomplete preference lists\footnote{Not all the agents are ranked. In the case of incomplete preference lists, the Rural Hospital Theorem holds~\cite{roth1984evolution,GaleSotomayor1985,roth1986allocation,manlove2013algorithmics}, which states that all the stable matchings have the same cardinality.} without ties or complete preference lists that include ties.\footnote{Some agents in the preference list are ranked equally. In the case of preference lists with ties, all the weakly stable matchings are complete (under the assumption that the cardinalities on the two sides of the bipartition are equal). Weakly stability means there is no pair of agents that strictly prefer to be matched to each other rather than being in their current assignment.} Once we assume both, incomplete lists and ties, the problem of finding the maximum cardinality stable matching becomes NP-hard even under very restrictive conditions \cite{manlove2002hard}. In terms of approximation ratios, the best known factor is $\frac{3}{2}$ \cite{kiraly2013linear} and the best lower bound is $\frac{33}{29}$ \cite{yanagisawa2007approximation}.


The design of a stable matching mechanism, when the number of participants of one side is increased, has already been investigated in the past. If hospitals have capacity one, this is known as the \emph{entry comparative static} in the stable marriage problem. In this setting, the two sides are traditionally called women and men. In~\cite{KelsoCrawford1982,GaleSotomayor1985,RothSotomayor1990}, the authors proved that when a new woman is added to the instance,  all men are matched \textit{weakly better}. 
Recently, Kominers \cite{Kominers2019} extended this result to the many-to-one stable matching problem. On a similar path, Balinski and Sonmez \cite{BalinskiSonmez1999} proved that the DA method is invariant with respect to residents who improve their score in the ranking lists, \ie, instead of introducing a new agent, the ranking of an existing agent is improved. 
A substantial part of the literature has focused on strategy-proof matching mechanisms, \ie, on matching mechanisms that incentivize participants to reveal their true preferences. Sonmez \cite{Sonmez1997} proved that hospitals can manipulate the stable matching in their favor by falsely reporting a reduced capacity. Moreover, Romm \cite{Romm2014} proved that the stable matching mechanism can still be manipulated even if the reported capacities are enforced during the admission process. Digressing from the entry comparative static approach and strategy-proof mechanisms, Bobbio et al.~\cite{bobbio2021capacity} considered the allocation of extra capacities to hospitals as a decision variable rather than a parameter. Alternatively to the integer programming approach proposed by the authors, a solution methodology for optimizing the outcome for the residents is devised in~\cite{abe2022anytime}. Our work focuses on providing the computational complexity landscape of the problem tackled in~\cite{bobbio2021capacity,abe2022anytime}, its counterpart where existent hospital spots are removed, and other variants.

In \cite{chade2006simultaneous,ali2021college}, the problem of capacity expansion was addressed in the framework of residents' interviews for the hospital admission and was solved through an optimal portfolio choice. In \cite{manjunath2021interview}, the authors studied how the expansion of interviews (only on the residents' side) impacts the final matching. 
Kamada and Kojima~\cite{kamada2015efficient} studied matching mechanisms that impose regional quotas for the Japan Residency Matching Program. Our work differs from this, since we look to optimize the quotas rather than imposing them. A problem related to ours was addressed in \cite{yahiro2020game}. As part of their problem's input, they considered a profile of ``resources'' that can be allocated to ``projects'' (which would be the hospitals). The authors concentrated on designing strategy-proof and efficient mechanisms. In the capacity variation problems considered in our work, the resources are decision variables rather than part of the input. Even if we translate the capacity expansion problem into the setting in~\cite{yahiro2020game}, the input size would be exponential.
In \cite{nguyen2018near}, the authors studied the presence of couples in matching markets. The authors proved that by adding at most 9 extra capacities in a market with couples, the existence of a stable matching is guaranteed. 

The problems studied in our work also relate to the literature on resource augmentation \cite{phillips1997optimal,kalyanasundaram2000speed,roughgarden2020resource}, where the goal is to design algorithms whose performance is compared to the benchmark that takes decisions with complete information but with a deficit in resources. For more details, we refer the interested reader to \cite{roughgarden2020resource} and the references therein.
Finally, in the context of ride sharing, it has been shown that the expansion of the capacities of the drivers by $\epsilon$ leads to a substantial reduction in the cost of the matching \cite{akbarpour2021value}, which, however, is not required to be stable.

\subsection{Contributions and Organization} 

This paper is organized as follows. 
In Section~\ref{sec:preliminaries}, we introduce the formal notation, the problem of expanding capacities (Problem~\ref{problem_def}), and the problem of reducing capacities (Problem~\ref{problem_reduction}). 

In Section~\ref{sec:expansion}, our main focus is to establish the complexity of Problem~\ref{problem_def}. To achieve this result, we first prove, in Corollary~\ref{corollary-Manlove}, that determining the resident-optimal stable matching in the presence of ties is NP-hard and is not approximable within $\bar{n}^{1-\varepsilon}$, for any $\varepsilon>0$, where $\bar{n}$ is the number of residents. This result puts a boundary on the computability of the resident-optimal stable matching, which is well known to be polynomially solvable when there are no ties.
The remainder of Section~\ref{sec:expansion} is devoted to study the complexity of the capacity expansion problem. All our results are proven in the special case in which the capacity of every hospital is at most 1.
 Indeed, even under very restrictive assumptions, finding the allocation of extra capacities to the hospitals that minimizes the average hospital rank of the residents is NP-hard, and for any $\varepsilon>0$, it cannot be approximated within a factor of $(\bar{n}/2)^{(1-\varepsilon)/2}$ 
 unless P=NP (Theorem~\ref{HR_ext_NP_complete}). This result may be counter-intuitive because, in the vanilla version of HR, we can compute in polynomial-time the resident-optimal stable matching, which is equivalent to the one that minimizes the average hospital rank. Our complexity proof is based on a new structure that we call \emph{village}. Each village is assigned some extra capacities, and the preferences of the hospitals and residents in a village ensure that the extra capacities can be optimally allocated only in a specific way. 

In Section~\ref{sec:reduction}, we study the capacity reduction problem. We prove that this problem is NP-hard, and for any $\varepsilon>0$, it cannot be approximated within a factor of $(\bar{n}/2)^{(1-\varepsilon)/2}$ 
unless P=NP (Theorem~\ref{theorem_reduction_standard}). The proof follows a similar reasoning as in Theorem~\ref{HR_ext_NP_complete}. We exploit again the structure of the village, in which now every relevant hospital has capacity 1. 



In Section~\ref{sec:extensions}, we study several variants of Problems \ref{problem_def} and \ref{problem_reduction}. Specifically, we partition the set of hospitals and allocate (remove) a certain amount of capacities to (from) each set of the partition. Theorem~\ref{theorem_stud_set_expansion} shows that, even when we partition the set of hospitals and we allocate to or remove from each set at most one spot, finding the optimal allocation is an NP-hard problem. Moreover, we prove that the optimization version of the problem is not approximable within a factor of $\bar{n}^{1-\varepsilon}$, for any $\varepsilon>0$ 
(Theorem~\ref{theorem_stud_set_expansion}). The equivalent results for the reduction problem are shown in Theorem \ref{capacity_red_sub_complexity}. Finally, we provide similar results to the variant of the problems that consider as an objective function the cardinality of the matching, Theorems \ref{theorem_max_expansion} and \ref{card_cap_sub_red_complexity}, respectively.

Finally, some conclusions are drawn in Section \ref{sec:conclusions}.
A summary of our results and relevant results from the literature can be found in Table~\ref{tab:our_results}.

\begin{table}[ht]
\begin{center}
\resizebox{\columnwidth}{!}{
\begin{tabular}{|l|cc|}\hline
& \multicolumn{2}{c|}{\textsc{Decision version of the problem}}\\
\textsc{Framework} &   Maximum cardinality & Average hospital rank    \\ \hline 
HR/HRI & Polynomial \cite{roth1984evolution,GaleSotomayor1985,roth1986allocation,manlove2013algorithmics} & Polynomial \cite{gale1962college,mcvitie1971stable,dubins1981machiavelli,roth1982economics} \\
HRT & Polynomial \cite{manlove2013algorithmics} & Inapprox. (Section~\ref{sec:expansion}) \\
HRTI & NP-complete  \cite{manlove2002hard} & Inapprox. (Section~\ref{sec:expansion}) \\ \hline
HR capacity variation & Trivial  & Inapprox. (Sections~\ref{sec:expansion} and \ref{sec:reduction}) \\
\hline
HR cap. variation subsets & Trivial & Inapprox. (Section~\ref{sec:extensions})\\
HRI cap. variation subsets & NP-complete (Section~\ref{sec:extensions}) & Inapprox. (Section~\ref{sec:extensions})\\\hline
\end{tabular}
}
\vskip 0.10in
\caption{\normalsize{Compilation of our contributions and relevant computational complexity results from the literature. HR corresponds to the many-to-one stable matching problem, the suffixes I and T stand for incomplete preference lists and for preference lists with ties, respectively.
}
}
\label{tab:our_results}
\end{center}
\end{table}

\section{Preliminaries and Problem Definition}\label{sec:preliminaries}
The many-to-one stable matching problem consists of a set of residents $\S=\{i_1,\ldots,i_{|\S|}\}$, a set of hospitals $\C=\{j_1,\ldots,j_{|\C|}\}$ and a set of edges $\E$ between $\S$ and $\C$. A resident and a hospital are linked by an edge in $\E$ if they deem each other acceptable. In this work, we assume (if not otherwise stated) that every resident-hospital pair is acceptable, \ie, $\E=\S\times\C$. Each hospital $j\in\C$ has an non-negative integer capacity $c_j\in\ZZ_+$ 
that represents the maximum number of residents that hospital $j$ can admit.
In this setting, a \emph{matching} $M$ is a subset of $\E$ in which each hospital $j$ appears in at most $c_j$ pairs and each resident appears in at most 1 pair. We denote by $M(i)$ and $M(j)$ the hospital assigned to resident $i$ and the subset of residents assigned to hospital $j$, respectively.

An instance $\Gamma$ of the HR problem corresponds to a tuple $\Gamma=\langle \S,\C,\succ, \mathbf{c} \rangle$, where $\mathbf{c}\in\ZZ_+^{ \C}$ is the vector of capacities and $\succ$ corresponds to the profile of preferences that residents have over hospitals and vice-versa. 
Specifically, we assume that the preference list of each resident is a linear order. We use the notation  $j\succ_i j'$ to describe when resident $i$ prefers hospital $j$ over hospital $j'$. 
We assume that every agent is {\it individually rational}, i.e., every agent prefers the proposed assignment than to be unmatched. 
Concerning the preference list of every hospital, we assume it is a \textit{responsive} linear order over the power-set of the residents \cite{roth1985college}.\footnote{For any two subsets of residents $\S',\S''$, we denote that hospital $h$ prefers $\S'$ over $\S''$ as $\S'\succ_{h} \S''$. A preference relation of a hospital is responsive if for every $\S'\subseteq \S$ with $|\S'|\leq c_h$, $s'\in \S'$ and $s'' \notin \S'$, we have that  (i) $\S' \succ_h \S'\cup \{s''\} \setminus \{s'\}$ if and only if  $\{ s'\} \succ_h \{s''\}$, and (ii) $\S' \succ_h \S'\setminus \{s'\}$ if and only if $\{ s'\} \succ_h \emptyset$. Therefore, a responsive preference list can be obtained from the linear order over singletons. Since responsive preferences are substitutable and satisfy the law of aggregated demand, our results hold also under these more relaxed assumptions.}
A responsive linear order facilitates the description of the preference list, since we only have to focus on the linear order over single residents.
We write $i\succ_{j} i'$ to denote when hospital $j$ prefers resident $i$ over $i'$. Whenever the context is clear, we drop the subscript in $\succ$. We emphasize that in the HR problem, unless otherwise stated, the preference lists are complete and strict (there are no \emph{ties}). Under these assumptions, the length of the preference list of each agent, hospital or resident, is exactly the size of the other side of the bipartition. Therefore, preference lists can be interpreted in terms of rankings. Formally, for each resident $i\in\S$ and hospital $j\in \C$, we denote by $\text{rank}_i(j)\in\{1,\ldots,|\C|\}$ the rank of hospital $j$ in the list of resident $i$.
This means, for example, that the most preferred hospital has the lowest ranking. 
Analogously, we define $\text{rank}_j(i)\in\{1,\ldots,|\S|\}$ for all $j\in\C, \ i\in\S$.

Given a matching $M$, we say that a pair $(i,j)\in \E$ is a \emph{blocking pair} if the following two conditions are satisfied: (1) resident $i$ is unassigned or prefers hospital $j$ over $M(i)$, and (2) $|M(j)|<c_j$ or hospital $j$ prefers resident $i$ over at least one resident in $M(j)$. The matching $M$ is said to be \emph{stable} if it does not admit a blocking pair.  
Gale and Shapley \cite{gale1962college} showed that every instance of the HR problem admits a stable matching that can be found in polynomial time by the deferred acceptance method, also known simply as the Gale-Shapley algorithm. In particular, this algorithm can be designed to prioritize the residents in the following sense: Let $M$ and $M'$ be two different stable matchings, we say that a resident $i$ \emph{weakly prefers} $M$ over $M'$ if $M(i)\succ_i M'(i)$ or $M(i)=M'(i)$.
Then, the DA algorithm can be adapted to compute the unique stable matching that is weakly preferred by all residents over all the other possible stable matchings. Such unique stable matching is called \emph{resident-optimal}. 

\paragraph{Notation} To ease the exposition, we avoid using the symbol $\succ$ when presenting a preference list, instead we simply separate agents by \lq\lq,\rq\rq\, and use the convention that the leftmost agents are the most preferred. For instance, we will represent the preference list $w\succ w' \succ w''$ as $w,w',w''$. Throughout this work, for a given integer $k\geq 1$, we use the shorthand $[k]:=\{1,\ldots,k\}$. Finally, otherwise stated, we use indices $i$ for residents and $j$ for hospitals.


%

\subsection{Problem Definition}

In this work, we focus on the stable matchings that minimize the \emph{average hospital rank}. Recall that we denote by $\text{rank}_i(j)$ the position of hospital $j$ in the list of resident $i$. 
The average hospital rank of a matching $M$ is defined as
\begin{equation}\label{eq:matching_cost}
\avg(M):=\sum_{(i,j)\in M} \text{rank}_i(j),
\end{equation}
where, to ease the exposition, we do not divide by the total number of hospitals. 
We consider Expression \eqref{eq:matching_cost} as the objective function, since a basic result states that a stable matching $M$ is resident-optimal if, and only if, it is a stable matching of minimum average hospital rank \cite{bobbio2021capacity}. 

In our first problem, proposed in~\cite{bobbio2021capacity}, we aim to improve the allocation of residents by increasing the capacity of the hospitals. For a non-negative vector $\mathbf{t}\in\ZZ_+^{\C}$, we denote by $\Gamma_{\mathbf{t}} = \langle \S,\C, \succ,\mathbf{c}+\mathbf{t} \rangle$ an instance of the {\CA} problem in which the capacity of each hospital $j\in\C$ is $c_j+t_j$. Observe that $\Gamma_\mathbf{0}$ corresponds to the original instance $\Gamma$ with no capacity expansion.
Formally, we define the capacity expansion problem as follows. 
\begin{problem}[\CAexp]
\label{problem_def}
\boxxx{
{\sc instance}: A \emph{\CA} instance $\Gamma=\langle \S,\C,\succ, \mathbf{c} \rangle$, a non-negative integer expansion budget $B\in\ZZ_+$,  and a target value $K \in \mathbb{Z}_+$.\\
{\sc question}: Is there a non-negative vector $\mathbf{t}\in\ZZ_+^{\C}$ and a matching $M_{\mathbf{t}}$ such that 
\[
\text{\emph{AvgRank}}(M_{\mathbf{t}}) \leq K,
\]
where $\mathbf{t}$ satisfies $\sum_{j\in \C}t_j\leq B$  and $M_{\mathbf{t}}$ is a stable matching in instance $\Gamma_{\mathbf{t}}$?
}
\end{problem}


Given parameters $B$ and $K$, Problem \ref{problem_def} aims to determine the existence of an allocation of $B$ extra spots through vector $\mathbf{t}$ such that there is a stable matching with an average hospital rank of at most $K$.

Throughout the paper, we assume that the total capacity of the hospitals is at least the total number of residents, \ie, $\sum_{j\in \C}c_j\geq |\S|$. If this assumption does not hold, we must define the cost of an un-assigned resident. A natural option is to add an artificial hospital with large capacity such that is ranked last by every resident. Therefore, un-assigned residents will be allocated in the artificial hospital whose rank is $|\C|+1$. Note that as a consequence of our assumption, $\sum_{j\in \C}c_j\geq |\S|$, there may be hospitals that do not fill their quota. 


In our second problem, we aim to find the reduction of the hospitals' capacities such that the final average hospital rank is the lowest possible, i.e., that has the least impact on the allocation of residents. As before, for a non-negative vector $\mathbf{t}\in\ZZ_+^{\C}$, we denote by $\Gamma_{-\mathbf{t}} = \langle \S,\C, \succ,\mathbf{c}-\mathbf{t} \rangle$ an instance of the {\CA} problem in which the capacity of each hospital $j\in\C$ is $c_j-t_j$.
Formally, we define our second problem as follows.
\begin{problem}[\CAred]
\label{problem_reduction}
\boxxx{
{\sc instance}:  A \emph{{\CA}} instance $\Gamma=\langle \S,\C,\succ, \mathbf{c} \rangle$, a non-negative integer reduction budget $B\in\ZZ_+$ such that $- B+\sum_{j\in \C}c_j \geq |\S|$ and a target value $K \in \mathbb{Z}_+$.\\
{\sc question}: Is there a non-negative vector $\mathbf{t}\in\ZZ_+^{\C}$ and a matching $M_\mathbf{t}$ such that 
\[
\text{\emph{AvgRank}}(M_{\mathbf{t}}) \leq K,
\]
where $\mathbf{t}$ satisfies $\sum_{j\in \C}t_j\geq B$ and $(c_j - t_j) \geq 0$ for every $ j \in \C$,  and $M_\mathbf{t}$ is a stable matching in instance $\Gamma_{-\mathbf{t}} $?
}
\end{problem}
Note that in Problem~\ref{problem_reduction}, we have the additional constraint that the capacity of every hospital should remain non-negative after removing spots, \ie, $c_j-t_j\geq 0$ for all $j\in\C$. We further assume that the sum of the reduced hospitals' capacities is greater or equal than the number of residents, \ie, $-B + \sum_{j\in \C}c_j\geq |\S|$. 
As in Problem \ref{problem_def}, if this assumption does not hold, we can transform the instance by adding an artificial hospital with a large capacity (which is ranked last in every resident's list) and by allowing the reduction of capacities to the original hospitals only.

\section{The Capacity Expansion Problem}\label{sec:expansion}


Our main result in this section establishes the computational complexity and inapproximability of Problem \ref{problem_def}. Denote by {\CAexpOpt} the optimization version of Problem \ref{problem_def}, \ie, the problem of finding the allocation of extra spots and the stable matching in the expanded instance that minimizes $\avg$. Formally, our main result is the following.
\begin{theorem}\label{HR_ext_NP_complete}
{\CAexp} is NP-complete. Moreover, for any $\varepsilon>0$, {\CAexpOpt} cannot be approximated within a factor of $(\bar{n}/2)^{(1-\varepsilon)/2}$, where $\bar{n}$ is the number of residents, unless \emph{P=NP}. 
\end{theorem}

To give some insights on the difficulty of Problem \ref{problem_def},
we first present an intuitive approach when $B=1$ and we show that it does not always provide an optimal solution.
In real life instances, certain hospitals may be ``more popular'' than others, namely, they are preferred by well-known voting methods such as Majority or Borda count \cite{zwicker2016introduction}.  
Thus, when $B=1$, a natural approach is to assign the additional spot to the hospital that is preferred by the majority or Borda count. 
However, as the following example shows, this is not necessarily optimal.

\paragraph{Counterexample for the Majority and Borda count}  Let $\S=\{i_1,i_2,i_3,i_4,i_5, i_6\}$ and $\C=\{j_1,j_2,j_3,j_4 \}$.  We assume that all hospitals have the same preference list: $i_1\succ i_2\succ\cdots\succ i_6$. Hospitals $j_1, j_2$ and $j_3$ have each capacity 1, and hospital $j_4$ has capacity 3. Resident $i_1$ ranks hospitals as $j_2\succ j_1\succ j_3\succ j_4$. Resident $i_2$ ranks hospitals as $j_2\succ j_3\succ j_1\succ j_4$. Resident $i_3$ ranks hospitals as $j_3\succ j_2\succ j_4\succ j_1$. Residents $i_4$, $i_5$ and $i_6$ rank hospitals as $j_1\succ j_4\succ j_3\succ j_2$. The resident-optimal stable matching is $ M = \{ (i_1,j_2), (i_2,j_3), (i_3,j_4), (i_4,j_1), (i_5,j_4),(i_6,j_4)\}$ with $\avg(M)=11$. Now, consider Problem \ref{problem_def} with $B=1$ and $K=9$. For this instance, an intuitive solution is allocating the extra spot to $j_1$, which is the most preferred hospital according to both Majority vote and Borda vote; the allocation of one extra capacity to $j_1$ is sub-optimal. Indeed, if we expand the capacity $c_{j_1}=1$ to $c_{j_1}=2$, then resident $i_5$ would be assigned to hospital $j_1$, which leaves an extra spot in hospital $j_4$. This solution reduces the average hospital rank by 1 unit and the resulting matching does not meet the target $K=9$. Instead, if we expand the capacity of $j_2$ to 2, then resident $i_2$ is admitted by hospital $j_2$, leaving an empty spot in hospital $j_3$ that is filled by resident $i_3$; the resulting matching has an average hospital rank of 9.\\

As the previous example shows, the allocation of one extra spot is not trivial when we try to solve it by just looking at the residents' preferences. 
However, we can still solve this problem in polynomial time by doing an exhaustive search in combination with the DA algorithm. To achieve this, we compute the resident-optimal stable matching using the DA mechanism in the instance $\Gamma_{\mathbf{t}}$ with $\mathbf{t}=\one_{j}$ for each $j\in \C$, where $\one_j\in\{0,1\}^{\C}$ is the indicator vector whose $j$-th component is 1 and the rest is 0. Once we obtain the cost for each $j\in\C$, we output the resident-optimal stable matching of minimum average hospital rank. Finally, we compare with our target $K$ to decide if such an allocation exists or not. Since the DA algorithm's runtime complexity is $O(|\S|\cdot |\C|)$ \cite{gale1962college}, then this exhaustive search runs in $O(|\S|\cdot |\C|^2)$. Whether this can be improved remains an open question.

To prove Theorem \ref{HR_ext_NP_complete}, we first study a variant of the \emph{egalitarian stable marriage} problem \cite{manlove2013algorithmics}. 
Formally, the stable marriage (SM) problem corresponds to the HR problem where $c_j= 1$ for all $j\in\C$. 
We use SMT to indicate the version of SM when ties are present in the preference lists. A tie appears when an agent allocates in the same position of the list two different participants of the opposite side. For example, if the preference list for resident $i$ is $j_4, (j_1, j_3), j_2$,\footnote{Throughout this paper, round brackets denote a tie.} then the rankings are $\text{rank}_i(j_4) = 1$, $\text{rank}_i(j) = 2$ for $j \in \{j_1,j_3\}$ and $\text{rank}_i(j_2) = 4$. For the SMT problem, stable matchings can be defined in several ways, but in this paper we consider \emph{weak stability} \cite{manlove2013algorithmics}. Formally, a matching $M$ is weakly stable if there is no pair such that both agents \emph{strictly} prefer each other over their allocation in $M$.
An egalitarian stable matching is a stable matching that minimizes the total sum of the rankings, \ie, $\sum_{(i,j)\in M} [\text{rank}_i(j)+\text{rank}_j(i)]$. Manlove et al.~\cite{manlove2002hard} proved that the problem of finding the egalitarian stable matching for SMT is not approximable within $\bar{n}^{1-\epsilon}$, for any $\epsilon>0$, unless $\text{P}=\text{NP}$, where $\bar{n}$ is the size of one side of the bipartition. For more details, we refer to Theorem 7 in \cite{manlove2002hard}.

Let us define the following variant of the egalitarian SMT problem.
\begin{problem}[\MinW]
\label{problem_minW}
\boxxx{
{\sc Instance}: An \emph{SMT} instance $\Gamma = \langle \S,\C,\succ, \mathbf{c} \rangle$ with $ c_j = 1$ for all $j \in \C$ and a target value $K \in \mathbb{\ZZ_+}$.  \\
{\sc Question}: Is there a weakly stable matching $M$ such that $\text{\emph{AvgRank}}(M) \leq K$?}
\end{problem}

We use {\MinWOpt} to denote the optimization version of {\MinW}, \ie, the problem of finding a weakly stable matching that minimizes $\avg$. Using the ideas in \cite{manlove2002hard}, we can obtain the following result for \MinW. 
\begin{corollary}\label{corollary-Manlove}
{\MinW} is NP-complete. Moreover, for any $\varepsilon>0$, {\MinWOpt} is not approximable within a factor of $\bar{n}^{1-\varepsilon}$, unless $\mathrm{P}=\mathrm{NP}$, where $\bar{n}=|\C|$. This result holds even if ties are only in one side, there is at most one tie per list, and each tie is of length two. 
\end{corollary}
For completeness, we provide the proof of this corollary in the Appendix.
Let us now provide a sketch of the steps to prove Theorem \ref{HR_ext_NP_complete}.
Given an instance $\Gamma$ of {\MinW}, we construct the following instance $\hat\Gamma$ of  {\CAexp}: For every hospital in $\Gamma$ that has ties in its preference list, we create a \textit{village} of residents and hospitals with different capacities and strict preferences. 
In Lemma \ref{lemma_first_complexity}, we prove that the construction can be done in polynomial time and it selects a special stable matching in the new instance. Let $M$ be the stable matching of minimum average hospital rank in {\MinW};  in Lemma \ref{lemma_second_complexity}, we prove that the stable matching $\hat{M}_{\mathbf{t}}$ in $\hat\Gamma$ is in fact the stable matching of minimum average hospital rank in {\CAexp}. 
%
%


%

\subsection{Design of the Instance}\label{sec:exp_reduction}

First, we observe that {\MinW} is NP-complete even if ties occur only among the preference lists of residents and in each preference list there is at most one tie of length 2, and it is positioned at the head of the list. For more details, we refer to Remark~\ref{remark-wlog} in the Appendix. Throughout this section, we assume that an instance of SMT satisfies these properties. Now, we introduce a polynomial transformation from such an instance of {\MinW} to an instance of {\CAexp}.

Let $\Gamma = \langle \S,\C,\succ, \mathbf{c} \rangle$ be an instance of {\MinW} such that $c_j = 1$ for all $j\in \C$ and $|\S|=|\C|=n$. Let $L\leq n$ be the number of residents with ties in their preference list.
The set of residents is partitioned in two sets $\S=\S'\cup \S''$ where $\S'$ is the set of residents with a tie of length two at the head of the preference list and $\S''$ is the set of residents with a strict preference list. Henceforth, we fix an ordering of the residents in $\S$ and denote $\S'=\{i_1,\ldots,i_L\}$ and $\S''=\{i_{L+1},\ldots, i_n\}$. Since preference lists are complete, observe that in any weakly stable matching every resident is matched.\footnote{This follows from the hypothesis that preference lists are complete and that the total capacity of the hospitals can accommodate all the residents. 
} 

In the following, we create an instance $\hat\Gamma = \langle \hat\S,\hat\C,\hat\succ, \hat{\mathbf{c}} \rangle$ of {\CAexp} with a specific target value and budget.  

\paragraph{Hospitals and residents.} First, we add a copy of the hospitals in $\C=\{j_1,\ldots,j_n\}$ and residents in $\S''=\{i_{L+1},\ldots,i_{n}\}$ 
We also introduce a set of hospitals $\C^0=\{j^0_1,\ldots,j^0_{n^2}\}= \C^{0,1}\cup\ldots\cup\C^{0,n}$ of size $n^2$ (where each $\C^{0,h}$ has size $n$), a set $\Z=\{z_1,\ldots,z_{n^3}\}$ of hospitals of size $n^3$, and a set $\X=\{x_1,\ldots,x_{n\cdot L}\}$ of hospitals of size $n\cdot L$. Recall that we index the residents in $\S'$ as $i_1,\ldots,i_L$. For every resident $i_\ell\in \S'$, where $\ell\in[L]$, we introduce additional hospitals and residents, to form a structure that we call \textit{village}, namely
\begin{itemize}
    \item a set of residents $\W_\ell=\{ w_{\ell,h} \}_{h\in[n]}$;
    \item a resident $y_\ell$;
    \item two sets of hospitals $\V_{\ell,0}=\{v^{0}_{\ell,h} \}_{h\in [n]}$ and $\V_{\ell,1}=\{v^{1}_{\ell,h} \}_{h\in [n]}$. Let $\V_\ell=\V_{\ell,0} \cup \V_{\ell,1}$.
\end{itemize}
We denote as $\B_{i_\ell}$ the village associated with resident $i_\ell\in\S'$ and $\V:= \bigcup_{\ell=1}^L \V_\ell$. 
In summary, we have $\hat\S = \S''\cup\{\W_\ell,y_\ell\}_{\ell\in[L]}$ and $\hat\C = \C \cup \C^0\cup \X\cup \V\cup \Z$.

\paragraph{Capacity vector.} Now, let us construct the capacity vector: For each hospital $v\in\V\cup\C^0\cup\Z$, we consider $\hat{c}_v = 0$; for every other hospital $j\in \C\cup\X$, we take $\hat{c}_j=1$.

\paragraph{Preference lists.} We now proceed to construct the preference lists in $\hat\Gamma$.


Given a resident $i_\ell\in \R'$ with $\ell\in[L]$, let $(j_{\sigma_1}, j_{\sigma_2}), j_{\sigma_3}\ldots $ be her ranking of the hospitals in the original instance $\Gamma$ (the parenthesis symbolizes the tie at the head of the list). We provide the preference lists of the residents and hospitals in village $\B_{i_\ell}$ with $\ell\in [L]$, namely 
\begin{subequations}
\begin{alignat*}{2} 
 w_{\ell,1} &: v^1_{\ell,1},\V_{\ell,0}\setminus \{v_{\ell,1}^0\}, \C_{n-1}^0, j_{\sigma_1}   ,  \Z, \ldots , \X    \\
 w_{\ell,2} &: v^1_{\ell,2},\V_{\ell,0}\setminus \{v_{\ell,2}^0\}, \C_{n}^0, j_{\sigma_2},   \Z,\ldots , \X    \\
 w_{\ell,h} &:  v^1_{\ell,h},\V_{\ell,0}\setminus \{v_{\ell,h}^{0}\},  \C_{(n-1)\cdot h}^0, j_{\sigma_h},   \Z,\ldots , \X  \hspace{2em}  h\in\{3,\ldots, n\}\\ 
 y_\ell &:  v_{\ell, 2}^{0}, v_{\ell, 1}^{ 0}, v_{\ell, 3}^{ 0}, \ldots,  v_{\ell, n}^{ 0} ,\Z,  \ldots , \X   \\
 v^1_{\ell,h} &: w_{\ell,h},  \ldots \hspace{17.5em}   \ h\in[n] \\
 v^0_{\ell,h} &: \W_\ell\setminus \{ w_{\ell,h} \}  , y_\ell,  \ldots \hspace{13.2em}  h\in[n] 
\end{alignat*}
\end{subequations}
where $\C_h^0=\{j_{1}^0,\ldots,  j_{h}^0  \}$. The purpose of positioning of set $\C_h^0$ in the preference lists of the residents $w_{\ell,h}$ for $h\in[n]$, is to ensure that we can mimic  $\text{rank}_{i_\ell}(j)$, for every $j\in\C$, of the original instance.  The symbol \lq\lq \ldots\rq\rq\ means that the remaining agents on the other side of the bipartition are ranked strictly and arbitrarily.

Now, we present the preference list of the copy of every hospital $j\in \H$ and every resident $i\in \R''$ in the new instance. 
 We modify the original preference list of $j \in \C$ by substituting every resident $i_\ell\in \R'$ (for $\ell\in[L] $) with resident $w_{\ell, r}$, where $r=\text{rank}_{i_\ell}(j)$ is the rank of $j$ in the list of $i_\ell$. If $j$ is ranked first by $i_\ell$, then we substitute $i_\ell$ with resident $w_{\ell, 1}$ when $j$ is the first hospital listed in the tie, otherwise with $w_{\ell, 2}$. Then, hospital $j$ ranks arbitrarily a strict ordering of the remaining residents.

Let $i\in \S''$ and let $j_{\sigma_1},\ldots, j_{\sigma_n} $ be her strict and complete preference list in $\Gamma$.  The preference list of $i$ in our new instance  $\hat\Gamma$ is 
\[
\C^{0,1}_{n-1},j_{\sigma_1},\C^{0,2}_{n-1},j_{\sigma_2},\ldots,\C^{0,n}_{n-1}, j_{\sigma_n}, \Z\ldots \]  

The preference lists of hospitals in $\X$, $\C^0$ and $\Z$ are arbitrary. The sole purpose of the hospitals in $\X$ is to ensure that there are sufficient capacities for all the residents. The scope of the set $\C^0$ is to help mimic the original ranking of the copies of the hospitals. The set of hospitals $\Z$ is introduced to make costly certain re-allocation of extra spots. The set $\V$ is used to leverage the stability and ensure that different allocations of extra spots yield sub-optimal results.

\paragraph{Target value and budget.} Given target $K$ and $L$  residents  with a tie in their list in an instance of {\MinW}, we define target $\bar{K}=n\cdot K + 2n\cdot L$ and budget $B=L\cdot n$ for {\CAexp}, where $L$ is the number of residents with a tie in their list.\\

Finally, note that $\hat\S = \S''\cup\{\W_\ell,y_\ell\}_{\ell\in[L]}$ and $\hat\C = \C \cup \C^0\cup \X\cup \V\cup \Z$. Therefore, the instance $\hat\Gamma$ consists of $((n-L)+  L \cdot (n+1) ) + (n+ n^2 + n\cdot L + L\cdot 2n + n^3)= n^3 + n^2 + 4n\cdot L + 2n$ residents and hospitals, which is $O(n^3)$; therefore the construction can be done in polynomial time.





\begin{remark}
Note that if no extra spots are assigned in our new instance, the set of hospitals $\X$ ensures that the residents are always matched.\footnote{This could be done also by adding a copy of themselves at the end of their list, which is usually referred to as individual rationality. Matching with oneself means being unassigned.}
Matching the residents to the hospitals in $\X$ leads to a higher average hospital rank. 
In the following section, we will prove that it is optimal to assign $L\cdot n$ extra capacities to the hospitals in $\V$, whose initial capacity is zero. 
\end{remark}

\subsection{Useful Lemmata}
For this section, recall that we are considering budget $B=L\cdot n$, where $L=|\S'|$ is the number of residents with a tie in their preference list.

\begin{lemma}\label{lemma_first_complexity}
For every weakly stable matching $M$ in $\Gamma$ with $\text{\emph{AvgRank}}(M)=K^M$, there is an allocation $\mathbf{t}$ respecting the budget $B=L \cdot n$ and a stable matching $\hat{M}_\mathbf{t}$ in $\hat{\Gamma}_\mathbf{t} = \langle \hat\S,\hat\C, \hat\succ,\hat{\mathbf{c}}+\mathbf{t} \rangle$ with $\text{\emph{AvgRank}}(\hat{M}_\mathbf{t})=n\cdot K^M + 2n\cdot L$.
%
\end{lemma}
\begin{proof}

Let $M$ be a (complete) weakly stable matching in $\Gamma$. Recall that $\S'=\{i_1,\ldots,i_L\}$ is the set of residents in $\Gamma$ with a single tie at the head of the list. Define the following set of indices in $M$: 
\[
\texttt{Idx}(M) = \{(\ell,r): \ r=\text{rank}_{i_\ell}(j), \ (i_\ell,j)\in M\cap(\S'\times\C)\}.
\]
The set $\texttt{Idx}(M)$ contains the information of the pairs $\S'\times\C$ that are matched in $M$.  Given $M$ and \texttt{Idx}, we now define the following sets that will be helpful in this proof. First, define the set of residents
\[
\W^M = \{w_{\ell, r}\in \W: \ (\ell,r)\in \texttt{Idx}(M)\},
\]
and define the sets of hospitals
\begin{align*}
\V^{M,0} &= \{v^0_{\ell,r}\in \V: \ (\ell,r)\in \texttt{Idx}(M) \},\\
\V^{M,1} &= \{v^1_{\ell,h}: \ (\ell,r)\in \texttt{Idx}(M), \ h\in[n]\setminus\{r\}\}.
\end{align*}
We now provide an allocation of extra spots $\mathbf{t}$ with a total budget $B=L\cdot n$ and a stable matching $\hat{M}_{\mathbf{t}}$ in $\hat\Gamma_{\mathbf{t}}$. 

\begin{itemize}
\item \emph{Allocation of extra spots.} 
We assign one extra position to each hospital in $\V^{M,0}\cup\V^{M,1}$. For the rest of the hospitals, we assign 0 extra capacity. We denote this allocation $\mathbf{t}$. Formally, we have
\[
t_u = \left\{\begin{matrix}1 & u\in \V^{M,0}\cup \V^{M,1}\\ 0 & \text{otherwise}\end{matrix}\right.
\]
Since $L= |\S'|$, all of the extra positions $B=L\cdot n$ were used. 

\item \emph{Matching.} For each $(\ell,r)\in\texttt{Idx}(M)$ with $j$ such that $r = \text{rank}_{i_\ell}(j)$ in $\Gamma$, we match the following pairs in $\hat{M}_{\mathbf{t}}$: $(w_{\ell,r},j)$, $(y_\ell, v^0_{\ell,r} )$, and $(w_{\ell,h}, v^1_{\ell,h} )$ for $h\in [n]\setminus \{ r \} $. Note that if $j$ is ranked first by $i_\ell$, the hospital is listed first or second in the tie. If $j$ is listed first, then $r=1$ and we match the pair $(w_{\ell,1},j)$, otherwise, $r=2$ and we match the pair $(w_{\ell,2},j)$.
For each $(i,j)\in M$ with $i\in \S''$, we match the pair $(i,j)$ in $\hat{M}_{\mathbf{t}}$, where $j$ is the corresponding copy in $\C$; recall that $\S''$ is the set of residents with a strict preference list. Formally, matching $\hat{M}_{\mathbf{t}}$ is as follows:
\begin{align*}
\hat{M}_{\mathbf{t}} =& \{(i,j): \ (i,j)\in M\cap(\S''\times\C)\}\\
& \cup \{(w_{\ell,r},j): \ r=\text{rank}_{i_\ell}(j), \ (i_\ell,j)\in M\cap(\S'\times\C)\}\\
& \cup \{(y_\ell, v^0_{\ell,r}): \ (\ell,r)\in \texttt{Idx}(M)\}\\
&\cup \{(w_{\ell,h},v^1_{\ell,h}): \ (\ell,r)\in \texttt{Idx}(M), \ h\in[n]\setminus\{r\}\}.
\end{align*}
\end{itemize}

Let us verify that $\hat{M}_{\mathbf{t}}$ is a stable matching $\hat{\Gamma}_{\mathbf{t}} $. First, note that residents $i \in\S''$ and hospitals $j\in \C$ cannot create blocking pairs because of their stability in $M$. Now, let us check the stability of the pairs in each village $\B_{i_\ell}$, where $i_\ell\in\S'$ with $\ell\in[L]$. Consider $j\in\hat{\C}$ and assume for now that $j$ is ranked first by $i_\ell$, \ie, $j$ is listed first or second in the tie: $j = j_{r}$ with $r = 1,2$. The pairs matched in village $\B_{i_\ell}$ are $(w_{\ell, r},j)$, $(y_\ell, v^0_{\ell,r})$ and $(w_{\ell,h}, v^1_{\ell,h})$ for $h\in[n]\setminus\{r\}$.
\begin{itemize}
\item The pair $(w_{\ell, r},j)$ is clearly stable; in fact, $w_{\ell, r}$ cannot be matched to any of the hospitals in $v^1_{\ell,r}$ and $\V_{\ell,0}\setminus \{v_{\ell,r}^0\}$ because they have capacity 0. If $r=2$, $w_{\ell, r} $ cannot be matched to any hospital in $\C^0_{n}$ because they all have capacity 0. Also, $j$ 
cannot create a blocking pair. Indeed, all the residents $w_{\ell', r'}$ ranked in its preference list before $w_{\ell, r}$ are matched to hospitals of the form $v^1_{\ell', r'}$ that they rank first. The case in which $r=1$ is analogous.
\item For $h\in[n]\setminus\{r\}$, $w_{\ell,h}$ ranks $v^1_{\ell,h}$ first and vice-versa, hence the $n-1$ pairs $(w_{\ell,h}, v^1_{\ell,h})$ are stable.
\item If $r=2$, then $y_\ell$ ranks $v^0_{\ell,r}$ first, and $v^0_{\ell,r}$ cannot be matched to any of the residents in $\W_\ell\setminus \{ w_{\ell,r} \}$ because of the previous point; therefore, pairs $(y_\ell, v^0_{\ell,r})$ are stable when $r =2$. If $r=1$, then  $y_\ell$ ranks $v^0_{\ell,r}$ second, and $y_\ell$ cannot be matched to $v^0_{\ell,2}$ because it has capacity 0. As before, $v^0_{\ell,r}$ cannot create a blocking pair with any of the residents in $\W_\ell\setminus \{ w_{\ell,1} \}$ because they are matched to their most preferred hospital. Therefore, the pair $(y_\ell, v^0_{\ell,r})$ is also stable when $r = 1$. 
\end{itemize}
The case in which $j$ is ranked third or more by $i_\ell$ is analogous to the case in which $j$ is of the form $j_{r}$ with $r =2$ in $i_\ell$'s preference list. Therefore, $\hat{M}_{\mathbf{t}}$ is a stable matching in $\hat{\Gamma}_{\mathbf{t}}$.


Next, we compute the average hospital rank in $M$ and $\hat{M}_{\mathbf{t}}$. In $M$, we can distinguish whether a resident is matched to a hospital ranked first or not, and we can distinguish if a resident is in $\S'$ or $\S''$. Let $K''$ be the average hospital rank of residents in $\S''$, $K^t$ be the average hospital rank of the residents in $\S'$ that are matched to a hospital in their ties, and $K^s$ be the average hospital rank of the residents in $\S'$ that are matched to a hospital they rank third or more. Note that $K^t$ is also the number of residents matched to a hospital they rank first. The average hospital rank of $M$ is $K^M=K'' +K^t + K^s$. We now show that $\avg(\hat{M}_{\mathbf{t}})$ is $n \cdot K'' + n \cdot K^t + n \cdot K^s + (n+1) \cdot (K^t + L - K^t  ) +  L\cdot(n-1)$: 
\begin{itemize}
\item The first term, $n \cdot K''$, is given by the contribution from the residents in $\S''$.
\item The second term is given by the pairs $(w_{\ell,r},j)$, $(y_\ell, v^0_{\ell,r} )$ (for $r=1$ or $r=2$) in the villages $\B_{i_\ell}$ of the residents $i_\ell$ in $M$ that are matched to a hospital they rank first. 
\item The third contribution is given by the pairs  $(w_{\ell,r},j)$, $(y_\ell, v^0_{\ell,r} )$ (for $r\geq 3$) in the villages $\B_{i_\ell}$ of the residents $i_\ell$ in $M$ that are matched to a hospital they rank third or more. 
\item The forth term, $L\cdot(n-1)$, is given by the pairs of the form $(w_{\ell,h}, v^1_{\ell,h})$ for $h\neq r$, of which there are $n-1$ in each of the $L$ villages.
\end{itemize}
If we rearrange the terms, we obtain $\avg(\hat{M}_{\mathbf{t}})= n \cdot K'' + n \cdot K^t + n \cdot K^s + (n+1) \cdot (K^t + L - K^t  ) +  L\cdot(n-1) =  n \cdot K^M + 2 n \cdot L   $.
\end{proof}

%
In the next result, we show that the allocation vector and the stable matching constructed in Lemma~\ref{lemma_first_complexity} correspond to the solution with the minimum average hospital rank, as long as the original matching is of minimum average hospital rank.
%

\begin{lemma}\label{lemma_second_complexity}
Consider a weakly stable matching $M$ in $\Gamma$ of minimum average hospital rank. Then, the allocation $\mathbf{t}$ and the stable matching $\hat{M}_{\mathbf{t}}$ constructed in Lemma~\ref{lemma_first_complexity} are the solutions of minimum average hospital rank for $\hat{\Gamma}$ when $B= L \cdot n$.
\end{lemma}
\begin{proof}
Let $M$ be a stable matching in $\Gamma$ of minimum average hospital rank. Recall the instance $\hat{\Gamma}$ constructed in Section~\ref{sec:exp_reduction}, the allocation 
\[
t_u = \left\{\begin{matrix}1 & u\in \V^{M,0}\cup \V^{M,1}\\ 0 & \text{otherwise,}\end{matrix}\right.
\]
and the matching 
\begin{align*}
\hat{M}_{\mathbf{t}} &= \{(i,j): \ (i,j)\in M\cap(\S''\times\C)\}\\
& \cup \{(w_{\ell,r},j): \ r=\text{rank}_{i_\ell}(j), \ (i_\ell,j)\in M\cap(\S'\times\C)\}\\
& \cup \{(y_\ell, v^0_{\ell,r}): \ (\ell,r)\in \texttt{Idx}(M)\}\\
&\cup \{(w_{\ell,h},v^1_{\ell,h}): \ (\ell,r)\in \texttt{Idx}(M), \ h\in[n]\setminus\{r\}\},
\end{align*}
constructed in Lemma~\ref{lemma_first_complexity}.
Denote by $\bar{K}= nK^M + 2nL $, which is the average rank of $\hat{M}_{\mathbf{t}}$ in $\hat{\Gamma}$. Now, we will prove that any other feasible allocation $\tilde{\mathbf{t}}$ with total budget $B=L\cdot n$ and any stable matching $\hat{M}_{\tilde{\mathbf{t}}}$ in the expanded instance $\hat{\Gamma}_{\tilde{\mathbf{t}}}$ have $\avg(\hat{M}_{\tilde{\mathbf{t}}})\geq\bar{K}$. 
\footnote{Note that the optimal solution may not use the entire budget, for instance, when every resident is matched with her top choice. However, we can always arbitrarily assign the remaining extra spots without affecting the final average hospital rank.}

Given allocation $\mathbf{t}$, note that it is not optimal to move one extra capacity from a hospital $v^s_{\ell,h}$ to a hospital $j\in \C\cup \C^0 \cup \X \cup \Z$ for $s=0,1$. Indeed,  $\X$ already has $B$ positions available, but since it is at the end of the preference list of every resident, it would be sub-optimal to match a resident to a hospital in it. Similarly, it would be sub-optimal to allocate an extra-capacity to $\Z$, since all the residents are already matched to a hospital they prefer to any hospital in $\Z$. 

Regarding the hospitals in $\C^0$, let us assume we move a capacity from $v^s_{\ell,h}$ to a hospital $h^0\in \C^0$ with $s=0,1$ and $h\in[n]$.  If $v^s_{\ell,h}$ was matched to $y_\ell$, then $y_\ell$ will be matched to some hospital ranked after $\Z$ with a rank of at least $n^3$, making the transfer sub-optimal. Otherwise, $v^s_{\ell,h}$ was matched to resident $w_{\ell,h}$; therefore, resident $w_{\ell,h}$ will be matched to a certain $v^0_{\ell,h'}$, which was previously matched to $y_\ell$; hence, the same reasoning just seen for $y_\ell$ applies, making the transfer of the extra spots sub-optimal. 

The additional cost in the average rank obtained by moving a position from $v^s_{\ell,h}$ to a hospital in $\C$ follows the same reasoning just outlined for $\C^0$. Therefore, this is also a sub-optimal re-allocation, and it is optimal to assign all the extra capacities to hospitals in $\V$. Consequently, $B$ residents are matched to hospitals in $\V$ and the remaining residents are matched to hospitals in $\C\cup \C^0$.

Given the fact above, we only have to focus on feasible allocations in hospitals that belong to $\V$. In the following, we analyze why a different allocation of extra capacities in $\V$ does not lead to a stable matching with a lower average hospital rank.
Since $M$ is a stable matching of minimum average hospital rank in $\Gamma$ and each resident in the new instance $\hat{\Gamma}$ ranks in the first $n^3$ positions only one hospital in $\C$, a re-allocation of extra capacities within the same village would result in a matching with a worse objective within the village. Therefore, the re-allocation of extra capacities that could improve the objective is the one obtained by transferring extra positions from one village to another village. 
Consider $i_\ell, i_{\ell'}\in\S'$. We now analyze the effects of moving one extra capacity from village $\B_{i_{\ell'}}$ to village $\B_{i_{\ell}}$. The reason why we are analyzing these transfers of extra capacities is because the corresponding residents are not necessarily matched with their top choice so their ranking and the overall average ranking may improve. 


\begin{itemize}
\item From $v^{0}_{\ell',r'}$ to $v^{0}_{\ell,2}$. Note that, by assumption, both residents $y_{\ell'}$ and $y_\ell$ are matched in $\hat{M}_{\mathbf{t}}$ to $v^{0}_{\ell',r'}$ and $v^{0}_{\ell,r}$ (for a certain $r\leq n$), respectively. If we move one extra position from $v^{0}_{\ell',r'}$ to $v^{0}_{\ell,2}$, then $w_{\ell,r}$ un-matches from $j$ ($r=\text{rank}_{i_\ell}(j)$) and matches to $v^{0}_{\ell,2}$, thus providing a reduction in the objective value between $n$ and $n^2$. The preference list of $v^{0}_{\ell,2}$ prevents $y_\ell$ to be matched to $v^{0}_{\ell,2}$ because $w_{\ell,r}$  is more preferred and $w_{\ell,r}$ prefers $v^{0}_{\ell,2}$ over $j$. As a consequence, $j$ is un-matched and is re-assigned to $y_{\ell'}$ (who was un-matched as a consequence of removing the extra spot of $v^{0}_{\ell',r'}$). By matching $y_{\ell'}$ to $j$, the objective value increases by at least $n^3$. Therefore, it is sub-optimal to move an extra capacity in this way. 

\item All the remaining cases follow a similar reasoning.
\end{itemize}

Therefore, there is no allocation $\tilde{\mathbf{t}}\neq \mathbf{t}$ and a stable matching $\hat{M}_{\tilde{\mathbf{t}}}$ with an objective value strictly lower than $\bar{K}$. 


\end{proof}

\subsection{{\CAexp} is NP-complete}\label{sec:main_theorem_reduction}
In the following, we prove the main result of this section, Theorem \ref{HR_ext_NP_complete}.
\begin{proof}[Proof of Theorem \ref{HR_ext_NP_complete}]
{\CAexp} is clearly in NP since given $\mathbf{t}$ and a matching $\hat{M}_{\mathbf{t}}$ in instance $\hat{\Gamma}_{\mathbf{t}}$, we can verify in polynomial time whether $\hat{M}_{\mathbf{t}}$ is stable, whether the budget $B=L\cdot n$ is satisfied and whether its objective value is less than the target value. We now show that {\CAexp} is NP-complete.

From Corollary \ref{corollary-Manlove}, we know that {\MinW} is NP-complete. 
Consider the reduction given in Section \ref{sec:exp_reduction}. In the constructed instance $\hat\Gamma$ of {\CAexp}, we set the budget to $B=L\cdot n$ and the target value to $\bar{K}=n\cdot K+2n\cdot L$, where $K$ is the target value of the instance of {\MinW}.


First, suppose that the answer to the instance of {\MinW} is NO, \ie, there is no weakly stable matching $M$ with an average hospital rank less or equal than $K$. Let $M$ be a weakly stable matching with minimum average hospital rank $K^M$; note that $K^M>K$. Next, we prove that there is no allocation of extra positions and a stable matching in the respective instance $\hat\Gamma$ of {\CAexp} with an objective value less or equal than $nK + 2nL$. Indeed, in Lemma \ref{lemma_first_complexity}, we show that there is an allocation $\mathbf{t}$ and a stable matching $\hat{M}_{\mathbf{t}}$ with $\avg(\hat{M}_{\mathbf{t}})=nK^M + 2nL$. In Lemma \ref{lemma_second_complexity}, we prove that these are the best solutions since $M$ is the optimal matching. Therefore, $nK^M+ 2nL> nK + 2nL = \bar{K}$, which means that the answer for the instance of {\CAexp} is also NO.

%

On the other hand, consider a YES instance of {\MinW}. Then, there is a weakly stable matching  $M$  with an average hospital rank of $K^M\leq K$. Therefore, the allocation $\mathbf{t}$ and the stable matching $\hat{M}_{\mathbf{t}}$ in $\hat\Gamma$ constructed in Lemma \ref{lemma_first_complexity} have an objective value of  $nK^M+2nL\leq nK+2nL=\bar{K}$. Hence, the instance of {\CAexp} has a YES answer. 

Let us prove now that, for any $\varepsilon$, {\CAexpOpt} is not approximable within a factor of $(\bar{n}/2)^{\frac{1-\epsilon}{2}}$, where $\bar{n}$ is the number of residents, unless P$=$NP.
Consider an instance $\Gamma$ of {\MinW} with $n$ residents and $L\geq1$ of them with a tie in their preference list. Let $M^{yes}$ and $M^{no}$ be the stable matchings of minimum average hospital rank for the cases in which the answer of the decision problem {\MinW} is YES and NO, respectively. Corollary \ref{corollary-Manlove} implies that, for any $\varepsilon>0$, $\avg(M^{no})\geq n^{1-\epsilon}\cdot\avg(M^{yes})$. Now, consider the reduction presented in Section \ref{sec:exp_reduction} from instance $\Gamma$ to an instance $\hat{\Gamma}$ of {\CAexp}. Lemma \ref{lemma_second_complexity} implies that there are allocations $\mathbf{t}$ and $\mathbf{t}'$, and matchings $\hat{M}^{yes}_{\mathbf{t}}$ and $\hat{M}^{no}_{\mathbf{t}'}$ for the respective YES and NO answers of {\CAexp} such that
\begin{align*}
\avg(\hat{M}^{yes}_{\mathbf{t}}) &= n\cdot\avg(M^{yes}) + 2nL\\
\avg(\hat{M}^{no}_{\mathbf{t}'}) &= n\cdot\avg(M^{no}) + 2nL.
\end{align*}
Recall that the reduction in Section \ref{sec:exp_reduction} constructs $\hat{\Gamma}$ with $\bar{n}:=|\hat{\S}|\leq 2n^2$ residents. Then, for any $\varepsilon>0$, we have
\begin{align*}
\frac{\avg(\hat{M}^{no}_{\mathbf{t}'})}{\avg(\hat{M}^{yes}_{\mathbf{t}})}&=\frac{n\cdot\avg(M^{no}) + 2nL}{n\cdot\avg(M^{yes}) + 2nL}\geq \frac{\avg(M^{no})}{\avg(M^{yes})}\\
&\geq n^{1-\epsilon}\geq \left(\frac{\bar{n}}{2}\right)^{\frac{1-\varepsilon}{2}},
\end{align*}
where the first inequality is because $f(x)=(a+x)/(b+x)$ is increasing when $b>a$. This completes the proof.
\end{proof}

Note that the proof can be slightly modified to obtain a similar result for HR with incomplete preference lists, as long as the condition $\sum_{j\in\C}c_j\geq |\S|$ is met (\eg, we can remove one hospital in $\X$ from the preference list of a resident).

\section{The Capacity Reduction Problem}\label{sec:reduction}





In this section, we focus on Problem \ref{problem_reduction} that looks for the reduction of capacities such that the residents' allocations are impacted the least.
Our main result establishes the computational complexity of this problem. Formally, our result is the following.
\begin{theorem}\label{theorem_reduction_standard}
{\CAred} is NP-complete. 
Moreover, for any $\varepsilon>0$, {\CAredOpt} cannot be approximated within a factor of $(\bar{n}/2)^{(1-\varepsilon)/2}$, where $\bar{n}$ is the number of residents, unless \emph{P=NP}.
\end{theorem}
\begin{proof}
We restrict our analysis to the case in which hospitals' capacities are all 1.  
Recall that Problem \ref{problem_reduction} assumes that reducing the capacities of hospitals does not leave any resident un-assigned. 

First, clearly {\CAred} is in NP, since for a given vector $\mathbf{t}$ and a matching $M_{\mathbf{t}}$, we can verify in polynomial time whether $\mathbf{t}$ satisfies the lower bound on the number of spots to be removed, whether $M_{\mathbf{t}}$ is stable in $\Gamma_{-\mathbf{t}}$ and if the target value is attained. Now, we concentrate on showing that the problem is NP-complete.

The rest of the proof follows the same idea than the proof of Theorem~\ref{HR_ext_NP_complete}.
 We build a reduction from an instance $\Gamma$  of {{\MinW}} into an instance $\hat{\Gamma}$ of {\CAred}.  We assume that $\Gamma$ satisfies $|\S|=|\C|=n$, ties occur only in residents' lists, and each of their preference list has at most one tie of length 2 positioned at the head of it. Recall also that we denoted by $\S'$ the set of residents with a tie in their preference list and by $\S''$ the set of residents with strict preference lists. The corresponding  $\hat\Gamma$ is defined as in the reduction presented in the proof of Theorem~\ref{HR_ext_NP_complete}, with the following difference:
\begin{itemize}
    \item  For every \emph{village} $\B_\ell$ defined for $i_\ell\in\S'$: Each hospital $v^1_{\ell,h}$ for $h\in[n]$ has capacity 1 and each hospital in $\V_{\ell,0}$ has capacity 1. All the remaining preferences and capacities remain as in Section \ref{sec:exp_reduction}.
\end{itemize}

%

Given a weakly stable matching $M$ in the instance $\Gamma$ with an average hospital rank $K^{M}$, we provide a reduction of the capacities $\mathbf{t}$ that respects the budget $B=n\cdot L$ and we build a stable matching $\hat{M}_{\mathbf{t}}$ in $\hat{\Gamma}_{\mathbf{t}}$ with an average hospital rank $\bar{K}=n K^M+  2nL$. 

\begin{itemize}
\item \emph{Reduction of capacities.} We remove $n$ spots from each village $B_\ell$ in the following way: Assume in $M$ we have the pair $(i_\ell, j)$, where $j$ is such that $r=\text{rank}_{i_\ell}(j)$. Then, we reduce by 1 the capacities of $v^1_{\ell,r}$  and of each hospital in $\V_{\ell,0}\setminus\{ v^0_{\ell,r}\}$. 
\item \emph{Matching.} We build the matching $\hat{M}_{\mathbf{t}}$ as follows: We match $(w_{\ell,r}, j)$, $(y_{\ell}, v^0_{\ell,r})$, $\{(w_{\ell,h}, v^1_{\ell,h})\}_{h\neq r}$. The remaining pairs are the same as in the proof of Lemma~\ref{lemma_first_complexity}.
\end{itemize}
%

The rest of the proof is analogous to the proofs of Lemma~\ref{lemma_first_complexity}, Lemma~\ref{lemma_second_complexity}, and Theorem~\ref{HR_ext_NP_complete}. 

 \end{proof}

Note that the proof can be slightly modified to obtain a similar result for HR with incomplete preference lists, as long as the condition $-B+\sum_{j\in\C}c_j\geq |\S|$ is met.

\section{Extensions}\label{sec:extensions}

In this section, we investigate the variants of Problems \ref{problem_def} and \ref{problem_reduction} where the decision-maker has budgets for different subsets of hospitals. In the remainder of this section, we say that $\P = \{\C_1,\ldots\C_q\}$ is a partition of the set of hospitals $\C$ if $\cup_{k\in[q]}\C_k= \C$ and $\C_k\cap\C_{k'}=\emptyset$ for all $k,k'\in[q]$ with $k\neq k'$.

\subsection{Allocating Extra Spots to a Partition of Hospitals}\label{sec: subsets of extra capacities}
We generalize Problem \ref{problem_def} to the setting where the set of hospitals is partitioned and we seek to find an allocation of extra spots such that each part has a specific budget. Formally, we study the following problem. 
\begin{problem}[\CAexpsub]
\label{problem_expansion_subsets}
\boxxx{
{\sc instance}:  A \emph{{\CA}} instance $\Gamma=\langle \S,\C,\succ, \mathbf{c} \rangle$, a partition $\P=\{\C_1,\ldots\C_q\}$ of $\C$, budget for each part $\{ B_k \in\ZZ_+  \colon \,  k \in [q]\}$, and a non-negative integer target value $K \in \mathbb{Z}_+$.\\
{\sc question}: Is there a non-negative vector $\mathbf{t}\in\ZZ_+^{\C}$ and a matching $M_\mathbf{t}$ such that 
\[
\text{\emph{AvgRank}}(M_{\mathbf{t}}) \leq K,
\]
where $\mathbf{t}$ is such that $\sum_{j\in\C_k}t_j\leq B_k$ for each $k\in[q]$ 
and $M_\mathbf{t}$ is a stable matching in instance $\Gamma_{\mathbf{t}} $?
}
\end{problem}

The next result can be directly obtained by considering a single set of hospitals in the partition, \ie, $q= 1$ and $\P = \C$, and by using Theorem \ref{HR_ext_NP_complete}.
\begin{corollary}\label{HR_extsub_NP_complete}
{\CAexpsub} is NP-complete. 
\end{corollary}
Denote by {\CAexpsubOpt} the optimization version of {\CAexpsub}, \ie, the problem of finding an allocation of extra capacities and a stable matching in the expanded instance of minimum average hospital rank. 
In the following result, we show the approximation complexity of {\CAexpsubOpt}.

\begin{theorem}\label{theorem_stud_set_expansion}
For any $\varepsilon>0$, {\CAexpsubOpt} is not approximable within a factor of $n^{1-\varepsilon}$, unless P=NP, where $n$ is the number of residents. This result holds even if the partition $\P=\{\C_1,\ldots,\C_q\}$ is such that each $\C_k$ contains at most two hospitals and $B_k\in \{0,1\}$ for every $k\in [q]$.
\end{theorem}

Before proving Theorem \ref{theorem_stud_set_expansion}, we need to introduce a variant of Problem \ref{problem_expansion_subsets}, where the goal is to find a stable matching whose size is at least a certain target. The problem of finding the maximum cardinality stable matching is one of  the main focus of the literature \cite{manlove2013algorithmics}. We investigate it in relation with capacity expansion when there are incomplete preference lists.
Recall that a {\CA} instance with incomplete preference lists means that there is at least one resident or one hospital that does not rank completely the opposite side. 
Formally, we consider the following problem.
\begin{problem}[{\Maxcardexp}]
\label{maxcard_expansion_subsets}
\boxxx{
{\sc instance}:  A \emph{HR} instance $\Gamma=\langle \S,\C,\succ, \mathbf{c} \rangle$ with incomplete preference lists, a partition $\P=\{\C_1,\ldots,\C_q\}$ of $\C$, budget for each part $\{ B_k \in\ZZ_+  \colon \,  k \in [q]\}$ and a non-negative integer target value $K \in \ZZ_+$.\\
{\sc question}: Is there a non-negative vector $\mathbf{t}\in\ZZ_+^{\C}$ and a matching $M_{\mathbf{t}}$ such that 
\[
\left|M_{\mathbf{t}}\right| \geq K,
\]
where $\mathbf{t}$ is such that $\sum_{j\in\C_k}t_j\leq B_k$ for each $k\in[q]$ 
and $M_\mathbf{t}$ is a stable matching in instance $\Gamma_{\mathbf{t}} $?
}
\end{problem}

Recall that if we consider complete preference lists, the problem above becomes trivial since all stable matchings have the same size. We prove the following result.
\begin{theorem}\label{theorem_max_expansion}
{\Maxcardexp} is NP-complete, even if the partition $\P=\{\C_1,\ldots,\C_q\}$ is such that each $\C_k$ is of size at most two and $B_k\in \{0,1\}$ for every $k\in[q]$.  
\end{theorem}
The proof of this result can be found in the Appendix. Let us now focus on the proof of Theorem \ref{theorem_stud_set_expansion}.

\begin{proof}[Proof of Theorem \ref{theorem_stud_set_expansion}]

Let $\epsilon>0$ and define $a=\lceil (3 / \varepsilon)\rceil $. We consider an instance $\Gamma$ of {\Maxcardexp} in which the set of hospitals is $\C$, the set of residents is $\S$ (we assume $|\S|=n$), every $\C_k$ is of size at most two and $B_k\in \{0,1\}$ for every $k\in [q]$. We denote by $O_j$ (resp. $O_i$) the preference list of hospital $j$ (resp. resident $i$). We assume that the target value $K$ is equal to $n$.

We now build an instance $\hat\Gamma$ of {\CAexpsubOpt}. Let us define $A=n^{a-1}$. In this instance, the set of hospitals is $\left(\bigcup_{h=1}^A \C^h \right) \cup \C^0$ with $\C^h=\{j_1^h,\ldots, j_n^h \}$ and $\C^0=\{j_1^0,\ldots, j_{n^a}^0 \}$. The set of residents is $\left(\bigcup_{h=1}^A \S^h\right) \cup \S^0$ with $\S^h=\{i_1^h,\ldots, i_n^h \}$ and $\S^0=\{i_1^0,\ldots, i_{n^a}^0 \}$.
For every pair $\{\C_k,B_k\}$ for $k\in[q]$ in $\Gamma$, we introduce copies in $\hat\Gamma$ of the form $\{\C_k^h,B_k^h\}$ in $\C^h$  for $k\in [q]$ and $h\in[A]$. The hospitals in $\C^0$ have all capacity 1; the other hospitals have the same capacities of the original hospitals in $\Gamma$. For $j\in\C$ and $h\in[A]$, we denote by $O_j^h$ the preference list  obtained by substituting in the preference list $O_j$ the residents in $\S$ with the residents in $\S^h$. We define similarly $O_i^h$. The preference lists of the hospitals and residents in $\hat\Gamma$ are as follows: 
\begin{align*}
j_{h}^{0}&: i_{h}^{0},\ldots \hspace{5.2em} h\in[n^{a}] \\
j_{s}^{h}&: O_{j_s}^{h}, \S^{0},\ldots \hspace{3em} s\in[n], \, h\in[A]\\
i_{h}^{0}&: j_{h}^{0},\ldots \hspace{5.2em} h\in[n^{a}] \\
i_{s}^{h}&: O_{i_s}^{h}, \C^{0}, \ldots \hspace{3em}  s\in[n], \, h\in[A],
\end{align*}
where the dots \lq\lq\ldots \rq\rq\, in the preference lists mean that the remaining agents on the other side of the bipartition are ranked strictly and arbitrarily.

Our {\SMexpsub} instance comprises $2 n^{a}$  residents, so that $\bar{n}:=2 n^{a} $; the target value  is $K'=n^{a+2}/2$.
The remainder of the proof follows the same reasoning as of Corollary~\ref{corollary-Manlove}; the proof can be found in the Appendix.
\end{proof}

\subsection{Removing Spots from a Partition of Hospitals}\label{sec: reduction_subsets}
Similar to the problems presented in the previous section, we now study the generalization of Problem \ref{problem_reduction} where the set of hospitals is partitioned in $q$ parts and each part has a budget for the removal of spots. Specifically, we consider the following problem.

\begin{problem}[{\CAredsub}]
\label{problem_reduction_subsets}
\boxxx{
{\sc instance}:  A \emph{{\CA}} instance $\Gamma=\langle \S,\C,\succ, \mathbf{c} \rangle$, a partition $\P=\{\C_1,\ldots\C_q\}$ of $\C$, budget for each part $\{ B_k \in\ZZ_+  \colon \,  k \in [q]\}$, and a non-negative integer target value $K \in \mathbb{Z}_+$.\\
{\sc question}: Is there a non-negative vector $\mathbf{t}\in\ZZ_+^{\C}$ and a matching $M_\mathbf{t}$ such that 
\[
\text{\emph{AvgRank}}(M_{\mathbf{t}}) \leq K,
\]
where $\mathbf{t}$ is such that $\sum_{j\in\C_k} t_j\geq B_k$ and $c_j-t_j\geq 0$ for $k\in[q]$,
and $M_{\mathbf{t}}$ is a stable matching in instance $\Gamma_{\mathbf{t}} $?
}
\end{problem}
For Problem \ref{problem_reduction_subsets}, we prove the following inapproximability result.
\begin{theorem}\label{capacity_red_sub_complexity}
For any $\varepsilon>0$, {\CAredsubOpt} is not approximable within a factor of $n^{1-\varepsilon}$, unless P=NP, where $n$ is the number of residents. This result holds even with a partition in which each part $\C_k$ contains at most two hospitals and $B_k\in \{0,1\}$ for every $k\in[q]$.
\end{theorem}

To prove Theorem \ref{capacity_red_sub_complexity}, we need to study the analogous version of Problem \ref{maxcard_expansion_subsets} for the capacity reduction setting. Formally, we define the following problem.

\begin{problem}[{\Maxcardred}]
\label{maxcard_reduction_subsets}
\boxxx{
{\sc instance}:  A \emph{{\CA}} instance $\Gamma=\langle \S,\C,\succ, \mathbf{c} \rangle$ with incomplete preference lists, a partition $\P=\{\C_1,\ldots\C_q\}$ of $\C$, budget for each part $\{ B_k \in\ZZ_+  \colon \,  k \in [q]\}$, and a non-negative integer target value $K \in \mathbb{Z}_+$.\\
{\sc question}: Is there a non-negative vector $\mathbf{t}\in\ZZ_+^{\C}$ and a matching $M_\mathbf{t}$ such that 
\[
\left|M_{\mathbf{t}}\right| \geq K,
\]
where $\mathbf{t}$ is such that $\sum_{j\in\C_k} t_j\geq B_k$ and $c_j-t_j\geq 0$ for $k\in[q]$,
and $M_{\mathbf{t}}$ is a stable matching in instance $\Gamma_{\mathbf{t}} $?
}
\end{problem}
In particular, we show the following result.
\begin{theorem}\label{card_cap_sub_red_complexity}
{\Maxcardred} is NP-complete. This result holds even with a partition in which each part $\C_k$ is of size at most two and $B_k\in \{0,1\}$ for every $k\in[q]$.
\end{theorem}
\begin{proof}
The proof is analogous to the proof of Theorem~\ref{theorem_max_expansion} with the difference that every hospital in each part $\C_k$ has capacity 1.
\end{proof}

\begin{proof}[Proof of Theorem~\ref{capacity_red_sub_complexity}]
The proof follows a similar reasoning as the proof of Theorem~\ref{theorem_stud_set_expansion} with the difference that every hospital in each part $\C_k$ has capacity 1.
\end{proof}


\section{Conclusions}
\label{sec:conclusions}

In this work, we have investigated the following question: How should a centralized institution optimally manage a variation in the capacities of the hospitals? 
We addressed this question from two points of view: Capacity expansion and capacity reduction. Our analysis is focused on the computational complexity of these problems and some of its variations. 

Our first result established the approximation hardness of the problem of finding the resident-optimal stable matching in the presence of ties. Our theorem defined a boundary on the complexity of the resident-optimal stable matching, which is well known to be polynomial-time solvable when there are no ties. We used this result as the first building block in the construction of the main proof of the paper: The approximation hardness of the problem of allocating optimally extra capacities to the hospitals to reduce the average hospital rank. Our proof introduced a crucial structure, the \emph{village}, that enabled us to manage the subtleties of the allocation of extra capacities. The problem of allocating extra resources is not easier when we restrict the distribution of capacities to a partition of the hospitals. If the objective of the problem is the cardinality of the stable matching, we proved that it is NP-complete when the problem has incomplete lists. If the objective is the average hospital rank, the corresponding optimization problem cannot be approximated within a certain factor.  The problem of reducing the capacities is equally interesting. Indeed, we showed that the capacity reduction problem is NP-complete. We generalized this result to the case in which the set of hospitals is partitioned and there is a budget for each part. For this problem, we proved that its optimization version is also inapproximable within a certain factor. Finally, we studied the variant of the problem that seeks to maximize the cardinality of the matching when the preference lists are incomplete. 

We believe these results are significant because they emphasize the existence of an underlying structure in the stable matching problem which governs both the capacity expansion and reduction. Unveiling the properties of this structure is certainly an open question worth being explored. Another interesting future direction of research is understanding what is the role of meta-rotations \cite{gusfield1989stable,bansal2007polynomial,cheng2008unified} in the capacity variation problem. 




\section*{Acknowledgments}
This work was funded by the Institut de valorisation des donn\'ees and Fonds de Recherche du Qu\'ebec through the FRQ–IVADO Research Chair in Data Science for Combinatorial Game Theory, and the Natural Sciences and Engineering Research Council of Canada through the discovery grant 2019-04557. Part of the work was conducted when the third author was Canada Excellence Research Chair at Polytechnique Montr\'eal, the generous support of the CERC grant being warmly acknowledged.

\bibliographystyle{plain}
{\small \bibliography{mybibfile}}

\appendix
\section{Missing proofs}\label{sec:appendix}
The following problem will be useful for the proofs that we provide in this Appendix. 
\begin{problem}[{\Maxcard}]
\label{problem_maxcar}
\boxxx{
{\sc instance}:  An \emph{HRTI} instance $\Gamma = \langle \S,\C,\succ, \mathbf{c} \rangle$ with $c_j\in\{0,1\}$, for all $j \in \C$, $| \C|=|\S|$ and a non-negative integer target value $K \in \ZZ_+$.\\
{\sc question}: Is there a weakly stable matching $M$ such that $\vert M \vert \geq K$?
}
\end{problem}
Recall that HRTI corresponds to the problem with ties and incomplete preference lists. In~\cite{manlove2002hard}, the authors proved that {\Maxcard} is NP-complete. As the next remark states, this result holds even if ties are at the head of the preference list, only on one side of it, at most one tie per list, and each tie is of length 2.

\begin{remark}\label{remark-wlog}
After the proof of Lemma 1 in \cite{manlove2002hard}, the authors showed that the problem {\Maxcard} 
can be simplified to the case in which ties are only on one side of the bipartition and are at the end of the preference list. Since the ties of the new instance created in Lemma 1 from \cite{manlove2002hard} are at most two, we can use the same reasoning to assume instead, without loss of generality, that in an instance of {\Maxcard} and the corresponding {\MinW} instance of Corollary~\ref{corollary-Manlove} ties occur only at the \emph{head} of a preference list.  
\end{remark}

\subsection{Proof of Corollary \ref{corollary-Manlove}}

In this section, we prove that {\MinW} is NP-complete and its optimization version cannot be approximated within a certain factor.  The proof is inspired by the proof of Theorem 7 in~\cite{manlove2002hard}. The result in~\cite{manlove2002hard} is stated in the traditional notation of the stable marriage problem where both sides are defined as women and men, instead of residents and hospitals. To keep coherence with the previous work, for this proof we also denote both sides as women and men.

\begin{proof}[Proof of Corollary \ref{corollary-Manlove}]
Clearly, {\MinW} is in NP.  Given $\varepsilon>0$, let $a= \lceil (3 / \varepsilon)\rceil .$ From Theorem 2 in \cite{manlove2002hard}, we know that, when ties occur on the women's side only, and each tie has length two, {\Maxcard} is NP-complete. Consider an instance of Problem~\ref{problem_maxcar} with $\C=\left\{m_{1}, m_{2}, \ldots, m_{n}\right\}$  and $\S=\left\{w_{1}, w_{2}, \ldots, w_{n}\right\}$. We assume that the target value $K$ is equal to $n$, since it was shown that even for this target value the problem is NP-complete. Let $O_{h}$ (resp. $R_{h}$) denote the preference list of man $m_{h}$ (resp. woman $w_{h}$) for $h\in[n]$. Next, we build an instance of {\MinW}. Let $C:=n^{ a-1 }$, then 
 \begin{itemize}
     \item the set of men is $\C'=\C^{0} \cup\left(\bigcup_{h=1}^{C} \C^{h}\right)$ with $ \C^{0}=\left\{m_{1}^{0}, m_{2}^{0}, \ldots, m_{n^a}^{0}\right\}$ and $\C^{h}=\left\{m_{1}^{h}, m_{2}^{h}, \ldots, m_{n}^{h}\right\}$ for $h\in[C]$;
     \item the set of women is $\S'=\S^{0} \cup\left(\bigcup_{h\in[C]} \S^{h}\right)$ with
    $\S^{0}=\left\{w_{1}^{0}, w_{2}^{0}, \ldots, w_{n^a}^{0}\right\}$ and $\S^{h}=\left\{w_{1}^{h}, w_{2}^{h}, \ldots, w_{n}^{h}\right\}$ for $h\in[C]$;
    \item for each $h\in[n]$ and $s\in[C]$, let $O_{h}^{s}$ be the preference list obtained from $O_{h}$ by replacing woman $w_{k}$ in $O_{h}$ by the corresponding woman $w_{k}^{s}$, for every  $k\in[n]$. We refer to the women in $O_{h}^{s}$ as the proper women for $m_{h}^{s} .$ Similarly, we define $R_{h}^{s}$ and the proper men for $w_{h}^{s}$. The preference lists for $\C'$ and $\S'$ are 
\begin{align*}
    m_{h}^{0}&: w_{h}^{0} \ldots  \hspace{4.7em} h\in[n^{a}] \\
    m_{h}^{s}&: O_{h}^{s}, \S^{0} \ldots \hspace{3em} h\in[n], \, s\in[C]\\
    w_{h}^{0}&: m_{h}^{0} \ldots  \hspace{4.6em} h\in[n^{a}] \\
    w_{h}^{s}&: R_{h}^{s},\C^0 \ldots \hspace{3em}  h\in[n], \ s\in[C]
\end{align*}
    where the dots \lq\lq\ldots \rq\rq\, in the preference lists mean that the remaining agents on the other side of the bipartition are ranked strictly and arbitrarily, and the sets mean that the agents within are ranked according to their indices;
    \item the target value is $K'=(n^{a+2})/2$.
 \end{itemize}
 
Our {\MinW} instance comprises $2 n^{a}$ men and $2 n^{a}$ women, so that $\bar{n}:=2 n^{a} $. Note also that the only ties in {\MinW} occur in the preference lists of women $w_{h}^{s}$ for $h\in[n], \ s\in[C]$. Moreover, there is at most one tie per list, and each tie has length 2. 

Suppose that we have a YES instance for {\Maxcard}, \ie, there is a stable matching $M$ with $|M|=n$. 
We create a matching $M^{\prime}$ in {\MinW} as follows: For every $h\in[n^a]$, we add the pair $\left(m_{h}^{0}, w_{h}^{0}\right)$ to $M^{\prime}$, and for each $s\in[n]$, we add the pair $\left(m_{s}^{\ell}, w_{k}^{\ell}\right)$ to $M^{\prime}$ for all $ \ell\in[C]$, where $\left(m_{s}, w_{k}\right) \in M $. Note that $M^{\prime}$ is
stable for our {\MinW} instance. We also have that
\[
\avg\left(M^{\prime}\right) \leq n^a + n^{a-1} n^{2} \leq \frac{n^{a+2}}{2}=K',
\]
since, without loss of generality, we may choose $n \geq 3$. Therefore, the objective value in  {\MinW} satisfies the target of $K'$.

On the other side, let us suppose that we have a NO instance for {\Maxcard}, \ie, it does not have a stable matching of cardinality $n$. Then, in any stable matching $M^{\prime}$ of {\MinW}, it holds that, for every $s\in[C]$, there is some $h\in[n]$ for which $w_{h}^{s}$ is not matched to one of her proper men. 
Nonetheless, in $M^{\prime}, m_{h}^{0}$ and $w_{h}^{0}$ must be partners, for every $h\in[n^{a}]$. Therefore, there is some $h\in[n]$ such that  $\text{rank}_{w_{h}^{s}}\left( M'(w_{h}^{s}) \right)>n^{a} $. 
Hence, $\avg\left(M^{\prime}\right)> n^{2 a-1}>K'$ for any stable matching of our {\MinW} instance. 

Therefore, the existence of a polynomial-time approximation algorithm for {\MinWOpt}  whose approximation ratio is as good as $\left(2 n^{2 a-1}\right) / n^{a+2}=2 n^{a-3}$ would give a polynomial-time algorithm for determining whether {\Maxcard} has a stable matching in which everybody is matched (\ie, $K=n$). To conclude, we note that $2 n^{a-3}=\left(2 / 2^{1-3 / a}\right) \bar{n}^{1-3 / a}>\bar{n}^{1-3 / a} > \bar{n}^{1-\varepsilon}$, which ends the proof.
\end{proof}

\subsection{Proof of Theorem \ref{theorem_max_expansion}}
In this section, we prove that {\Maxcardexp} is NP-complete.

\begin{proof}[Proof of Theorem \ref{theorem_max_expansion}]
We build  a polynomial reduction from an instance of {\Maxcard} where ties are only on the hospital side, they are at the head of the preference list and are of length two. Let $\C$ and $\S$ be the set of hospitals and residents in $\Gamma$, respectively; $\C=\C'\cup \C''$, where $\C'$ is the set of hospitals with a tie at the head of the preference list and $\C''$ is the set of hospitals with a strict preference list. 

We build an instance $\hat\Gamma=\langle \hat\S,\hat\C,\hat\succ, \hat{\mathbf{c}} \rangle$ of {\Maxcardexp} as follows: 
\begin{itemize}
    \item The set of residents is a copy of $\hat\S=\S$;
    \item The set of hospitals $\hat\C$ consists of a copy of $\C''$ and the set $\tilde \C = \{ j': \ j \in \C'\} \cup \{j'': \ j\in \C' \}$, \ie, we make two copies per hospital in $\C'$. Each hospital in $\tilde \C$ has capacity 0 and each hospital in $\C''$ has capacity 1; 
    \item For each resident in $\hat\S$, we keep the preference list that she has in the original instance $\Gamma$, with the exception that each $j \in \C'$ in her preference lists is replaced by $j'$ if she does not appear in the tie. If she is the first resident listed in the tie of $j\in \C'$, then we replace the hospital $j$ in the preference list by $j'$; otherwise, if the resident is listed second in the tie of $j\in \C'$, then we replace the hospital $j$ in the preference list by $j''$;
    \item For the hospitals in $\C''$, we maintain their preference lists of $\Gamma$ over the residents in $\S$. For a hospital $j\in \C'$ with a preference list $(i_{\sigma_1}, i_{\sigma_2}),i_{\sigma_3},\ldots,i_{\sigma_s}$,  the preference list of $j'$ becomes $i_{\sigma_1},i_{\sigma_3},\ldots,i_{\sigma_s}$ and that of $j''$ becomes $i_{\sigma_2},i_{\sigma_3},\ldots,i_{\sigma_s}$;
    \item For each hospital $j\in \C''$, we create a set $\C_j=\{j\}$ with $B_j =0$. For every hospital $j\in \C'$, we create a set $\C_j=\{ j', j''\}$ with $B_j =1$. Clearly, the sets $\C_j$ induce a partition of the set of hospitals $\hat\C$. 
    \item The target value is $K$, \ie, the same as in the {\Maxcard} instance.
\end{itemize}




Let $M$ be a weakly stable matching of the {\Maxcard} instance. We will show that there is a feasible allocation of the capacities $\mathbf{t}$ and a stable matching $M_{\mathbf{t}}$ in $\hat{\Gamma}_{\mathbf{t}}$ with the same cardinality, and thus, establishing the problems equivalence.  For every pair $(i,j)$ in $M$, we have to distinguish whether $j\in \C'$ or $j\in \C''$. If $j\in \C''$, then we add the corresponding pair $(i,j)$ to $M_{\mathbf{t}}$; recall that for a hospital $j\in\C''$, $B_j = 0$. 
Otherwise, $j\in \C'$. If $i\neq i_{\sigma_2}$, then we allocate the extra capacity of part $\C_j$ to $j'$ and we match the pair $(i,j')$. If, instead, $i= i_{\sigma_2}$, then we match the pair $(i,j'')$ by assigning the extra capacity of part $\C_j$ to $j''$. If there is a hospital $j\in \C'$ that has not been assigned to any resident, then we may allocate the extra capacity of part $\C_j$ to $j'$.\footnote{Alternatively, we may leave un-assigned the extra capacity $B_j$ for every unassigned hospital $j\in \C'$ in $M$.}

Note that $M_{\mathbf{t}}$ is stable indeed. If not, there must be a blocking pair $(i,j)$ where both the resident and the hospital have a capacity of 1 (otherwise, a hospital with capacity 0 could not create a blocking pair). Note that $j$ must be in some $\C_k$ given that those subsets form a partition of $\hat\C$.  Indeed, in each set $\C_k$  exactly one hospital has capacity 1, and for $k=j$, $j$ is exactly such hospital. If $|\C_j|=1$, then $j\in\C''$ and, thus, it has exactly the same preference list that it has in the instance $\Gamma$; therefore the corresponding pair $(i,j)$ in $M$ is a blocking pair, which yields a contradiction. If $|\C_j|=2$, then we have to distinguish whether $j=j'$ or $j=j''$. If $j=j'$, then we find that $(i,j)$ is a blocking pair in $M$. Otherwise, if $j=j''$, then $(i,j)$ is a blocking pair if and only if $i= i_{\sigma_2}$ since it is the only resident ranking $j''$ in $\hat\Gamma$.  The pair $(i,j'')$ could be a blocking pair only if $j''$ has capacity 1; the extra capacity $B_j=1$ was assigned to $j''$ in accordance with the reduction. Therefore $(i,j'')$ is already matched in $M_{\mathbf{t}}$ and $(i,j'')$ cannot be a blocking pair. 

Note that we have created a bijection between the set of stable matchings in the {\Maxcard} instance  and the allocation of extra spots as well as the set of stable matchings in the {\Maxcardexp} instance modulo the  stable matchings in  the {\Maxcardexp} instance that have some unassigned hospitals of the form $j'$ or $j''$. Moreover, this correspondence preserves the cardinality of the stable matching. 

To conclude, note that the created instance introduces a polynomial number of hospitals, residents, preferences and pairs $\{(\C_j, B_j)\}_{j\in \C}$ in the input. Moreover, it can be verified in polynomial time that: (1) the vector of allocation $\mathbf{t}$ satisfies the corresponding constraints and (2) the constructed stable matching has a cardinality greater or equal than the target value. 
\end{proof}

\end{document}